\pdfoutput=1
\documentclass[a4paper, USenglish, cleveref]{lipics-v2021}
\usepackage[algo2e, ruled, vlined, linesnumbered]{algorithm2e}
\usepackage{blindtext}
\usepackage{booktabs}
\usepackage[group-four-digits, table-parse-only, table-alignment=right]{siunitx}
\usepackage{tikz}

\DontPrintSemicolon
\SetKwProg{Proc}{Procedure}{}{end}
\SetKwProg{Func}{Function}{}{end}
\title{Nearest-Neighbor Queries in Customizable Contraction Hierarchies and Applications}
\titlerunning{Nearest-Neighbor Queries in CCHs and Applications}

\author
    {Valentin Buchhold}
    {Karlsruhe Institute of Technology, Germany \and
        \url{https://i11www.iti.kit.edu/en/members/valentin_buchhold/index}}
    {buchhold@kit.edu}{}{}
\author
    {Dorothea Wagner}
    {Karlsruhe Institute of Technology, Germany \and
        \url{https://i11www.iti.kit.edu/en/members/dorothea_wagner/index}}
    {dorothea.wagner@kit.edu}{}{}

\authorrunning{V. Buchhold and D. Wagner}
\Copyright{Valentin Buchhold and Dorothea Wagner}



\supplementdetails{Software}{https://github.com/vbuchhold/routing-framework}

\hideLIPIcs
\nolinenumbers

\EventEditors{David Coudert and Emanuele Natale}
\EventNoEds{2}
\EventLongTitle{19th International Symposium on Experimental Algorithms (SEA 2021)}
\EventShortTitle{SEA 2021}
\EventAcronym{SEA}
\EventYear{2021}
\EventDate{June 7--9, 2021}
\EventLocation{Nice, France}
\EventLogo{}
\SeriesVolume{190}
\ArticleNo{18}

\newcommand*\Dist{\mathit{dist}}
\newcommand*\SepDecomp{\mathcal{T}}
\newcommand*\Rank{\pi^{-1}}

\newcommand*\Key{\mathit{key}}
\newcommand*\RevLabel{d_\mathsf{r}}
\newcommand*\RevQueue{Q_\mathsf{r}}

\newcommand*\Otot{O_\text{tot}}
\newcommand*\Ofit{O_\text{fit}}
\newcommand*\Osel{O_\text{sel}}
\newcommand*\Oint{O_\text{int}}

\begin{document}

\maketitle

\begin{abstract}
  Customizable contraction hierarchies are one of the most popular route planning frameworks in practice, due to their simplicity and versatility. In this work, we present a novel algorithm for finding $k$-nearest neighbors in customizable contraction hierarchies by systematically exploring the associated separator decomposition tree. Compared to previous bucket-based approaches, our algorithm requires much less target-dependent preprocessing effort. Moreover, we use our novel approach in two concrete applications. The first application are \emph{online} $k$-closest point-of-interest queries, where the points of interest are only revealed at query time. We achieve query times of about 25~milliseconds on a continental road network, which is fast enough for interactive systems. The second application is travel demand generation. We show how to accelerate a recently introduced travel demand generator by a factor of more than 50 using our novel nearest-neighbor algorithm.
\end{abstract}

\section{Introduction}

Motivated by route planning in road networks, the last two decades have seen intense research on speedup techniques~\cite{BastDGMPSWW16} for Dijkstra's shortest-path algorithm~\cite{Dijkstra59}, which rely on a slow preprocessing phase to enable fast queries. Particularly relevant to real-world production systems are customizable speedup techniques, which split preprocessing into a metric-independent part, taking only the network structure into account, and a metric-dependent part (the \emph{customization}), incorporating edge weights (the \emph{metric}). A fast and lightweight customization is a key requirement for important features such as real-time traffic updates and personalized metrics. The most prominent customizable techniques are \emph{customizable route planning} (CRP)~\cite{DellingGPW17} and \emph{customizable contraction hierarchies} (CCHs)~\cite{DibbeltSW16}. Both achieve similar performance but with different trade-offs, and both are in use in industry.

Modern map-based services must support not only point-to-point queries but also many other types of queries. Over the years, both CRP and CCHs have been extended to numerous types of queries and problems. Efentakis and Pfoser~\cite{EfentakisP14} propose one-to-all and one-to-many algorithms within the CRP framework, and Efentakis et al.~\cite{EfentakisPV15} extend CRP to nearest-neighbor queries. Delling and Werneck~\cite{DellingW15} present alternative CRP-based algorithms for the one-to-many and nearest-neighbor problem. Baum et al.~\cite{BaumDPW13} extend CRP so that it can find energy-optimal paths for electric vehicles, and Kobitzsch et al.~\cite{KobitzschRS13} so that it can find multiple alternate routes from the source to the target.

Customizable contraction hierarchies and in particular \emph{contraction hierarchies} (CHs)~\cite{GeisbergerSSV12}, the predecessors of CCHs, have also received considerable attention; see \cite{BastDGMPSWW16} for a recent overview. Since each CCH \emph{is a} CH, all algorithms operating on CHs carry over to CCHs. Delling et al.~\cite{DellingGNW13} introduce PHAST, a one-to-all algorithm on CHs. RPHAST~\cite{DellingGW11} is an extension to the one-to-many problem. Alternatively, one-to-many queries on CHs can be solved using the \emph{bucket-based approach} by Knopp et al.~\cite{KnoppSSSW07}. Geisberger~\cite{Geisberger11} extends the bucket-based approach to the nearest-neighbor problem.

In this work, we introduce a novel algorithm for finding $k$-nearest neighbors in CCHs. The \emph{$k$-nearest neighbor problem} takes as input a graph~$G = (V, E)$, a source~$s \in V$, a nonempty set~$T \subseteq V$ of targets, and an integer~$k$ with $1 \le k \le |V|$. The goal is to find the $k$ targets~$t_i \in T$ closest to $s$, i.e., those that minimize $\Dist(s, t_i)$, where $\Dist(v, w)$ is the shortest-path distance from $v$ to $w$. Modern nearest-neighbor algorithms tailored to road networks work in up to four phases~\cite{DellingW15, DellingGW11, AbeywickramaCT16}. \emph{Preprocessing} takes as input only the network structure, \emph{customization} incorporates the metric into the preprocessed data, \emph{selection} (or \emph{target indexing}) incorporates the set of targets into the data, and \emph{queries} take a source and find the $k$ targets closest to the source. Our algorithm follows this standard four-phase setup.

Note that there is already a nearest-neighbor algorithm by Geisberger~\cite{Geisberger11} which operates on CHs. However, its relatively heavy selection phase makes it only suitable for \emph{offline} queries, where the set of targets is known in advance. This is the case for simple store locators of franchises. However, more common in interactive map-based services are \emph{online} queries, where the set of targets is only revealed at query time. An example is the computation of the closest businesses whose name contains a user-defined keyword. We are not aware of any CH-based algorithm that can solve such queries.

There is indeed an algorithm~\cite{DellingW15} for online nearest-neighbor queries within the CRP framework. As already mentioned, however, CRP and CCHs are on a par with each other and both used in industry with good reasons. For a production system based on the CCH framework, it is usually not desirable to simultaneously maintain a CRP setup to support nearest-neighbor queries. All types of queries should be solvable within the CCH framework.

\subparagraph*{Related Work.}

We start by briefly reviewing the CH- and CRP-based nearest-neighbor algorithms mentioned above. Contraction hierarchies (CHs)~\cite{GeisbergerSSV12} are a point-to-point route planning technique that is much faster than Dijkstra's algorithm (four orders of magnitude on continental networks). CHs replace systematic exploration of \emph{all} vertices in the network with two much smaller search spaces (forward and reverse) in directed acyclic graphs, in which each edge leads to a ``more important'' vertex.

The basic idea behind the bucket-based nearest-neighbor algorithm~\cite{Geisberger11} is to precompute and store the reverse CH search spaces of the targets during the selection phase. More precisely, if $v$ appears in the reverse search space from a target~$t$ with distance~$y$, then $(t, y)$ is stored in a \emph{bucket}~$B(v)$ associated with $v$. The bucket entries are sorted by nondecreasing distance. The query phase of the bucket-based nearest-neighbor algorithm computes the forward CH search space from the source~$s$. For each vertex~$v$ in the search space from $s$ with distance~$x$, we scan the bucket~$B(v)$. For each entry~$(t, y) \in B(v)$, we obtain an $s$--$t$ path of length~$x + y$. The algorithm maintains the $k$~closest targets seen so far and stops bucket scans when $x + y$ reaches the distance to the $k$-th closest target found so far.

Customizable route planning (CRP)~\cite{DellingGPW17} is a point-to-point route planning technique that splits preprocessing into a metric-independent part and a metric-dependent customization. Metric-independent preprocessing partitions the network into roughly balanced cells and creates \emph{shortcuts} between each pair of boundary vertices in the same cell. Customization assigns costs to the shortcuts by computing shortest paths within each cell. Queries run a modification of bidirectional search that uses the shortcuts to skip over cells that contain neither the source nor the target. For better performance, we use multiple levels of overlays.

The CRP-based nearest-neighbor algorithm~\cite{DellingW15} marks all cells that contain one or more targets during the selection phase. Queries run a modification of Dijkstra's algorithm that skips over unmarked cells and descends into marked cells. Since the search discovers targets in increasing order of distance, we can stop when the $k$-th target is reached.

Of course, there are also nearest-neighbor algorithms tailored to road networks that are based on neither CRP nor CHs. Arguably the simplest one is \emph{incremental network expansion} (INE)~\cite{PapadiasZMT03}, which runs Dijkstra's algorithm until the $k$-th target is reached. Another straightforward approach is \emph{incremental Euclidean restriction} (IER)~\cite{PapadiasZMT03}. The basic idea behind IER is to repeatedly retrieve the next closest target based on the straight-line distance (e.g., using an R-tree~\cite{Guttman84}) and compute the actual distance to it using any shortest-path algorithm as a black box. IER stops when the geometric distance to the next closest target exceeds the shortest-path distance to the $k$-th closest target so far encountered.

Since IER had only been evaluated using Dijkstra's algorithm, its performance was generally regarded as uncompetitive in practice. In particular, IER combined with Dijkstra cannot possibly be faster than INE. More recently, IER was combined with \emph{pruned highway labeling}~\cite{AkibaIKK14}, yielding one of the fastest nearest-neighbor algorithms in many cases~\cite{AbeywickramaCT16, AbeywickramaC17}.

More sophisticated nearest-neighbor algorithms are SILC~\cite{SametSA08, SankaranarayananAS05}, ROAD~\cite{LeeLZT12, LeeLZ09}, and G-tree~\cite{ZhongLTZG15, ZhongLTZ13}. Since previously published results had disagreed on the relative performance of these algorithms, Abeywickrama et al.~\cite{AbeywickramaCT16} carefully reimplemented and reevaluated them once more. While G-tree was faster than SILC and ROAD in most cases, the differences were relatively small. Delling and Werneck~\cite{DellingW15} compare the CRP-based nearest-neighbor algorithm to G-tree, claiming that CRP outperforms G-tree. To sum up, all algorithms have comparable performance, with selection and query times of the same order of magnitude. However, a big advantage of CRP (and also of our algorithm) compared to the other approaches is a fast and lightweight customization phase, enabling important features such as real-time traffic updates and personalized metrics.

\subparagraph*{Our Contribution.}

We introduce a novel algorithm for finding $k$-nearest neighbors that operates on CCHs. Our algorithm systematically explores the associated separator decomposition tree in a way similar to nearest-neighbor queries~\cite{FriedmanBF77} in kd-trees~\cite{Bentley75}. Its selection phase is orders of magnitude faster than the one of previous bucket-based approaches, which makes it a natural fit for online $k$-closest point-of-interest (POI) queries. On the road network of Western Europe, we achieve selection times of about 20~milliseconds and query times of a few milliseconds or less. This enables \emph{interactive} online queries, which need to run both the selection and query phase for each client's request. We are not aware of any other nearest-neighbor algorithm operating on CCHs that enables interactive online queries.

In addition to closest-POI queries, we also look at a second concrete application. We show how a slightly modified version of our nearest-neighbor algorithm can be used for travel demand generation (or mobility flow prediction). Here, the problem we consider is computing the number~$T_{vw}$ of trips between each pair~$(v, w)$ of vertices~$v, w$ in a road network. Depending on the expected length of the generated trips, we accelerate a recently introduced demand generator~\cite{BuchholdSW19a} by a factor of more than 50.

\subparagraph*{Outline.}

\Cref{sec:preliminaries} reviews the CCH framework. \Cref{sec:nearest-neighbors} describes our novel nearest-neighbor algorithm in detail. \Cref{sec:applications} continues with two concrete applications in which our algorithm can be used. \Cref{sec:experiments} presents an extensive experimental evaluation of various closest-POI algorithms and travel demand generators. \Cref{sec:conclusion} concludes with final remarks.

\section{Preliminaries}
\label{sec:preliminaries}

We treat a road network as a bidirected graph $G = (V, E)$ where vertices represent intersections and edges represent road segments. Each edge $(v, w)$ has a nonnegative length $\ell(v, w)$ that represents the travel time from $v$ to $w$. A one-way road segment from $v$ to $w$ can be modeled by setting $\ell(w, v) = \infty$. The shortest-path distance (i.e., travel time) from $v$ to $w$ in $G$ is denoted by $\Dist(v, w)$. For simplicity, we assume that $G$ is strongly connected.

\subparagraph*{Separator Decompositions.}

A \emph{separator decomposition}~\cite{BauerCRW16} of a strongly connected $n$-vertex bidirected graph~$G = (V, E)$ is a rooted tree~$\SepDecomp = (\mathcal{X}, \mathcal{E})$ whose nodes~$X \in \mathcal{X}$ are disjoint subsets of $V$ and that is recursively defined as follows. If $n = 1$, then $\SepDecomp$ consists of a single node~$X = V$. If $n > 1$, then $\SepDecomp$ consists of a root~$X \subseteq V$ that separates $G$ into multiple strongly connected subgraphs~$G_0, \dots, G_{d - 1}$. The children of $X$ are the roots of separator decompositions of $G_0, \dots, G_{d - 1}$. For clarity, an element~$v \in V$ is always called \emph{vertex} and an element~$X \in \mathcal{X}$ is always called \emph{node}. We denote by $\SepDecomp_X$ the subtree of $\SepDecomp$ rooted at $X$ and we denote by $G_X$ the subgraph of $G$ induced by the vertices contained in $\SepDecomp_X$. The vertex set of $G_X$ is represented by $V(G_X)$, and the edge set by $E(G_X)$.

In general, a separator~$X \in \mathcal{X}$ should be small, and the resulting subgraphs~$G_0, \dots, G_d$ should be balanced. Therefore, separator decompositions are typically obtained by recursive dissection (e.g., using Inertial Flow~\cite{SchildS15}, FlowCutter~\cite{HamannS18}, or InertialFlowCutter~\cite{GottesburenHUW19}). 

\subparagraph*{Nested Dissection Orders.}

A separator decomposition~$\SepDecomp$ of $G$ induces a (not necessarily unique) \emph{nested dissection order}~$\pi$ on the vertices in $G$~\cite{George73}. To obtain one, we number the vertices in the order in which they are visited by a postorder tree walk of $\SepDecomp$, where the vertices in each node are visited in any order. Note that the resulting order~$\pi = \langle \pi_0, \dots, \pi_{d - 1}, \pi_d \rangle$ is split into $d + 1$ contiguous subsequences $\pi_i$, where $d$ is the number of children~$Y_j$ of the root~$X$ of $\SepDecomp$. The subsequences~$\pi_0, \dots, \pi_{d - 1}$ are nested dissection orders on $V(G_{Y_0}), \dots, V(G_{Y_{d - 1}})$, and $\pi_d$ is an arbitrary order on $X$. (In the presence of turn costs, picking $\pi_d$ carefully improves performance~\cite{BuchholdWZZ20}.) We denote by $\pi^{-1}(v)$ the \emph{rank} of $v$ in $\pi$.

\subparagraph*{Customizable Contraction Hierarchies.}

\emph{Customizable contraction hierarchies} (CCH)~\cite{DibbeltSW16} are a three-phase speedup technique to accelerate point-to-point shortest-path computations. The preprocessing phase computes a separator decomposition of $G$, determines an associated nested dissection order on the vertices in $G$, and \emph{contracts} them in this order. To contract a vertex $v$, it is temporarily removed, and \emph{shortcut} edges are added between its neighbors. The output of preprocessing is the input graph plus the shortcuts added during contraction. We call this graph $H$. We denote by $N_H^\uparrow(v)$ the set of neighbors of $v$ in $H$ ranked higher than $v$.

The customization phase computes the lengths of the edges in $H$ by processing them in bottom-up fashion. To process an edge $(u, w)$, it enumerates all triangles~$\{v, u, w\}$ in $H$ where $v$ has lower rank than $u$ and $w$, and checks whether the path~$\langle u, v, w \rangle$ improves the length of $(u, w)$. Alternatively, Buchhold et al.~\cite{BuchholdSW19b} enumerate all triangles~$\{u, w, v'\}$ in $H$ where $v'$ has higher rank than $u$ and $w$, and check whether the path~$\langle v', u, w \rangle$ improves the length of $(v', w)$, which accelerates customization by a factor of 2.

There are two query algorithms. First, one can run a bidirectional Dijkstra search on $H$ that only relaxes edges leading to vertices of higher ranks. Let a \emph{forward CCH search} be a Dijkstra search that relaxes only outgoing upward edges, and a \emph{reverse CCH search} one that relaxes only incoming downward edges. A \emph{CCH query} runs a forward CCH search from the source and a reverse CCH search from the target until the search frontiers meet.

In addition, there is a query algorithm based on the \emph{elimination tree} of $H$. The parent of a vertex~$v$ in the elimination tree is the lowest-ranked vertex in $N_H^\uparrow(v)$. Bauer et al.~\cite{BauerCRW16} prove that the ancestors of a vertex~$v$ in the elimination tree are exactly the set of vertices scanned by a Dijkstra-based CCH search from $v$. An elimination tree search from $v$ therefore scans all vertices in the CCH search space of $v$ in order of increasing rank by traversing the path in the elimination tree from $v$ to the root. Since elimination tree queries use no priority queues, they are usually faster than Dijkstra-based CCH queries.

\section{Our Nearest-Neighbor Algorithm}
\label{sec:nearest-neighbors}

\begin{algorithm2e}[tb]
  \caption{Recursive formulation of our nearest-neighbor algorithm. At the first call, the parameter~$X$ is the root of the separator decomposition tree.}
  \label{algo:nearest-neighbors}
  \Func{$\mathit{searchSepDecomp}(X)$}{
    \If{the recursion threshold is deceeded}{
      examine all targets $t \in T \cap V(G_X)$ in the subgraph $G_X$\;
      \Return\;
    }
    examine all targets $t \in T \cap X$ in the separator $X$\;
    $C \gets \emptyset$\;
    \ForEach{child $Y$ of $X$}{
      \If{$T \cap V(G_Y) \ne \emptyset$}{
        \eIf{$s \in V(G_Y)$}{
          $C \gets C \cup \{(Y, 0)\}$\;
        }{
          compute the distance $\Dist(s, Y)$ from $s$ to a closest vertex in $G_Y$\;
          $C \gets C \cup \{(Y, \Dist(s, Y))\}$\;
        }
      }
    }
    \ForEach{$(Y, \Dist(s, Y)) \in C$ in ascending order of $\Dist(s, Y)$}{
      \If{$\Dist(s, Y)$ is less than the distance to the $k$-th closest target seen so far}{
        $\mathit{searchSepDecomp}(Y)$\;
      }
    }
  }
\end{algorithm2e}

Our algorithm for finding nearest neighbors in CCHs is inspired by the algorithm of Friedman et al.~\cite{FriedmanBF77} for finding nearest neighbors in kd-trees~\cite{Bentley75}. (However, our description requires no knowledge of that algorithm.) During the search, we maintain the $k$ closest targets seen so far in a max-heap~$\hat{T}$ using their distances from the source as keys. Initially, $\hat{T} = \{\bot\}$ with $\Key(\bot) = \infty$. The basic idea is as follows: We systematically explore the separator decomposition tree, but visit only nodes~$X$ whose corresponding subgraph~$G_X$ contains vertices that are closer to the source than the $k$-th closest target found so far. For each visited node~$X$, we compute the shortest-path distance from the source to each target in the separator~$X$ (if any), and update $\hat{T}$ accordingly.

The precise algorithm is most easily formulated as a recursive procedure (see \cref{algo:nearest-neighbors}). It takes a node~$X$ in the separator decomposition tree as parameter. At the first call, $X$ is the root of the separator decomposition. The first step of the procedure is to \emph{examine} all targets~$t \in T \cap X$ in the separator~$X$. To examine a target~$t$, we compute the shortest-path distance~$\Dist(s, t)$ from $s$ to $t$ with a standard elimination tree search. If $\Dist(s, t)$ is less than the maximum key in $\hat{T}$, we insert $t$ into the heap. If $\hat{T}$ now contains $k + 1$ elements, we delete the maximum element from the heap and discard it.

Next, we loop over all children~$Y$ of $X$ in the separator decomposition tree. If the subgraph~$G_Y$ induced by the vertices in $\SepDecomp_Y$ contains any targets, we add a pair~$(Y, \Dist(s, Y))$ to a set~$C$. We denote by $\Dist(s, Y)$ the shortest-path distance from $s$ to a closest vertex in $G_Y$, i.e., $\Dist(s, Y) = \min_{v \in V(G_Y)} \Dist(s, v)$. If $G_Y$ contains the source, this distance is zero. Otherwise, we have to compute it, which we will discuss in the next sections.

Finally, we loop over all pairs~$(Y, \Dist(s, Y)) \in C$ in ascending order of distance from the source. If $\Dist(s, Y)$ is less than the distance to the $k$-th closest target seen so far, we recurse on $Y$. Otherwise, $\SepDecomp_Y$ cannot contain better solutions than those already known.

Note that when $G_X$ is large but contains only a few targets, it is less costly to loop over all these targets than to explore $\SepDecomp_X$ until the leaves are reached. Therefore, when the number of targets in $G_X$ drops below a certain threshold, we stop the recursion and examine all targets~$t \in T \cap V(G_X)$ in $G_X$ (we use a recursion threshold of 8 in our experiments, determined experimentally). The following sections work out the remaining details.

\subparagraph*{Accessing Vertices and Targets in Subgraphs.}

Given a node~$X$ in the separator decomposition tree, our algorithm requires easy access to the set of vertices and the set of targets in the subgraph~$G_X$ and in the separator~$X$. Accessing the set of vertices in $G_X$ and in $X$ is particularly easy. To improve cache efficiency, the vertices in a CCH are reordered according to the order of contraction. That is, the vertices are numbered in the order in which they are visited by a postorder tree walk of $\SepDecomp$, where the vertices in each node are visited in any order. Hence, for each $X \in \mathcal{X}$, the vertices in $G_X$ are numbered contiguously, with the vertices in $G_X \setminus X$ appearing before the vertices in $X$. To support easy access to the vertices in $G_X$ and in $X$, we only need to store three indices with each $X$: the vertex in $G_X$ with the smallest index, the vertex in $G_X$ with the largest index, and the vertex in $X$ with the smallest index.

The set~$T$ of targets is represented by a sorted array. To make the targets in subgraphs (or separators) easily accessible, we use an auxiliary array~$A$ of size $|V| + 1$. The element~$A[i]$, $0 \le i \le |V|$, stores the number of targets among the first $i$ vertices. Note that $A$ can be filled by a single sweep through $T$ and $A$. To access the targets in $G_X$ (or $X$), we first retrieve the index~$l$ of the first vertex and the index~$r$ of the last vertex in $G_X$ (or $X$), as discussed above. The number of targets in $G_X$ (or $X$) is then $A[r + 1] - A[l]$, and the actual targets are stored contiguously in $T[A[l]]$, \dots, $T[A[r + 1] - 1]$.

\subparagraph*{Computing Shortest Paths to Subgraphs.}

The most straightforward approach to compute the distance~$\Dist(s, X)$ from $s$ to a closest vertex in $G_X$ is a standard Dijkstra-based CCH query, where the reverse search is initialized with all vertices in $G_X$. Let $\RevLabel$ and $\RevQueue$ be the distance labels and the queue of the reverse search, respectively. To initialize the reverse search, we set $\RevLabel[v] = 0$ for each vertex~$v \in V(G_X)$, $\RevLabel[w] = \infty$ for each vertex~$w \in V \setminus V(G_X)$, and $\RevQueue = V(G_X)$. This yields a correct but inefficient algorithm. However, we can do better.

We define the \emph{boundary}~$B(X)$ of $G_X$ as the set of vertices in $V \setminus V(G_X)$ that are adjacent to a vertex in $G_X$, i.e., $B(X) = \{w \in V \setminus V(G_X): (v, w) \in E, v \in V(G_X)\}$. Note that the boundary of any $G_X$ is easily accessible without any additional preprocessing.

\begin{lemma}
  \label{thm:boundary-vertex-ranks}
  Let $u$ be any vertex in $G_X$. Then, $\Rank(b) > \Rank(u)$ for each $b \in B(X)$.
\end{lemma}

\begin{proof}
  Consider any vertex~$b \in B(X)$. Let $b$ be contained in the node~$Y \notin V(\SepDecomp_X)$. We claim that $Y$ lies on the path in $\SepDecomp$ from $X$ to the root~$R$. Assume otherwise, i.e., $Y$ does not lie on the $X$--$R$ path. Let $Z$ be the lowest common ancestor of $X$ and $Y$. Since $Z$ separates $G_X$ and $G_Y$ in $G_Z$, there is no edge in $G$ that connects $G_X$ and $G_Y$. This contradicts that $b$ is adjacent to a vertex in $G_X$. Thus, $Y$ lies on the $X$--$R$ path. Since the vertices are numbered in the order in which they are visited by a postorder tree walk of $\SepDecomp$, the vertices in $Y$ are assigned higher ranks than the ones in $X$. In particular, we have $\Rank(b) > \Rank(u)$.
\end{proof}

\begin{theorem}
  \label{thm:boundary-neighborhood}
  Let $u$ be the highest-ranked vertex in $G_X$. Then, $B(X) = N_H^\uparrow(u)$.
\end{theorem}

\begin{proof}
  Let $b$ be a vertex in $B(X)$. We claim that $b \in N_H^\uparrow(u)$. Since $G_X$ is by definition connected, there is a path~$\langle u, v_0, \dots, v_k, b \rangle$ in $G$ with $v_i \in G_X$. Since $\Rank(v_i) < \Rank(u)$ by definition and $\Rank(u) < \Rank(b)$ by \cref{thm:boundary-vertex-ranks}, all $v_i$ are contracted before $u$ and $b$. Therefore, CCH preprocessing adds a shortcut~$(u, b)$, and thus $b \in N_H^\uparrow(u)$.

  Conversely, let $w$ be a vertex in $N_H^\uparrow(u)$, i.e., there is an edge~$(u, w)$ in $H$. Since $u$ is the highest-ranked vertex in $G_X$ and $\Rank(u) < \Rank(w)$, we have $w \in V(G) \setminus V(G_X)$. We claim that $w \in B(X)$. Assume otherwise, i.e., $w \in V \setminus (V(G_X) \cup B(X))$. Since $B(X)$ separates $u$ and $w$ in $G$, the shortcut~$(u, w)$ corresponds to a path~$\langle u, \dots, b, \dots, w \rangle$ in $G$ with $b \in B(X)$. By construction, $b$ is contracted before $u$ and $w$. This contradicts \cref{thm:boundary-vertex-ranks}.
\end{proof}

If $s \in V(G_X)$, then $\Dist(s, X) = 0$. So, assume $s \notin V(G_X)$. Since $B(X)$ separates $s$ and $G_X$, and all edge lengths are nonnegative, there is a closest vertex~$v^*$ in $G_X$ such that there is a shortest $s$--$v^*$ path~$\langle s, \dots, b, v^* \rangle$, $b \in B(X)$. Note that $(b, v^*)$ is a shortest edge among all edges~$(b, v) \in E$, $v \in V(G_X)$; otherwise, $v^*$ would not be a closest vertex in $G_X$. Therefore, $\Dist(s, X) = \min_{b \in B(X)} (\Dist(s, b) + \min_{\{(b, v) \in E: v \in V(G_X)\}} \ell(b, v))$. That is, it suffices to initialize the reverse search of the query with all boundary vertices. More precisely, we set $\RevLabel[b] = \min_{\{(b, v) \in E: v \in V(G_X)\}} \ell(b, v)$ for each vertex~$b \in B(X)$, $\RevLabel[w] = \infty$ for each vertex~$w \in V \setminus B(X)$, and $\RevQueue = B(X)$. This yields a reasonable algorithm, but we can do even better by exploiting elimination tree queries, which are usually faster than the Dijkstra-based CCH queries we have used so far.

Recall that the CCH search space~$S(b)$ of a vertex~$b$ corresponds to the path in the elimination tree from $b$ to the root~$r$. An elimination tree search from $b$ therefore scans all vertices in $S(b)$ in order of increasing rank by traversing the $b$--$r$ path in the elimination tree. Given a set~$B$ of vertices, it is not clear how to enumerate all vertices in the union of the search spaces, since the union generally corresponds to a subtree rather than a path in the elimination tree. However, we can exploit the fact that in our case $B$ is the boundary of $G_X$.

\begin{theorem}
  \label{thm:search-spaces}
  Let $l$ be the lowest-ranked vertex in $B(X)$. Then, $S(l) = \bigcup_{b \in B(X)} S(b)$.
\end{theorem}

\begin{proof}
  Since $l \in B(X)$, we trivially have $S(l) \subseteq \bigcup_{b \in B(X)} S(b)$, so let $b \ne l$ be a vertex in $B(X)$. We claim that $S(b) \subseteq S(l)$. By \cref{thm:boundary-neighborhood}, the highest-ranked vertex~$u$ in $G_X$ is adjacent to both $l$ and $b$. Since $\Rank(u) < \Rank(l) < \Rank(b)$, CCH preprocessing adds a shortcut~$(l, b)$ when $u$ is contracted. Therefore, we have $b \in S(l)$ and thus $S(b) \subseteq S(l)$.
\end{proof}

By \cref{thm:search-spaces}, we can compute $\Dist(s, X)$ with a standard elimination tree query from $s$ to the lowest-ranked vertex in $B(X)$, where we initially set $\RevLabel[b] = \min_{\{(b, v) \in E: v \in V(G_X)\}} \ell(b, v)$ for each vertex $b \in B(X)$. Since a lower bound on $\Dist(s, X)$ suffices to preserve the correctness of our nearest-neighbor algorithm, we can also initialize the distance labels to zero. The resulting lower bound is only slightly worse than the exact distance, but initialization is somewhat faster. We observed the lowest running times when using lower bounds.

\subparagraph*{Accelerating Shortest-Path Searches.}

Note that the forward searches of all elimination tree queries done during the same nearest-neighbor query start at the same source. Unless we use special pruning criteria~\cite{BuchholdSW19b}, the forward searches compute identical distance labels. To further accelerate our nearest-neighbor algorithm, we run the forward search \emph{once} before the systematic exploration of the separator decomposition tree. Whenever we compute the distance to a target or subgraph, we run only the reverse search, which accesses the precomputed distance labels of the forward search.

After scanning a vertex $v$, a standard elimination tree search immediately initializes the distance label of $v$ to $\infty$, since it is not accessed anymore afterwards. We maintain this initialization approach for the reverse searches. The forward search, of course, must not immediately initialize the distance labels. Instead, after the exploration of the separator decomposition tree, we traverse the path in the elimination tree from the source to the root once again, and initialize the forward distance label of each visited vertex.

\section{Applications}
\label{sec:applications}

We continue with two substantially different applications in which our nearest-neighbor algorithm can be used. An obvious application are $k$-closest POI queries in map-based services. We can use our nearest-neighbor algorithm as is for this application, without further modifications. Afterwards, we look at a more abstract application (travel demand generation) where we make slight modifications to our algorithm.

\subsection{Online Closest-POI Queries}

Recall that modern closest-POI algorithms~\cite{DellingW15, DellingGW11, AbeywickramaCT16} work in up to four phases: preprocessing, customization, selection, and queries. We now divide the work our nearest-neighbor algorithm does into these standard phases. Note that our nearest-neighbor algorithm does nothing else but the standard CCH preprocessing and customization during the first two phases. To support easy access to the set of vertices in a subgraph or separator, we indeed need to associate three indices with each node~$X \in \mathcal{X}$ but an efficient representation of the separator decomposition already stores this information. Therefore, we reuse the standard CCH preprocessing and customization, without further modifications.

The selection phase runs POI-dependent preprocessing. The only preprocessed data that depends on the set~$P$ of POIs is the auxiliary array~$A$, which makes the POIs in a subgraph or separator easily accessible. As already mentioned, $A$ can be filled by a single sweep through $P$ and $A$. Finally, the query phase runs the systematic exploration of the separator decomposition tree (including the forward search immediately before the exploration and the initialization of the forward distance labels immediately after the exploration).

Note that our selection phase is lightweight and (as our experiments will show) orders of magnitude faster than the one of previous bucket-based approaches. This makes our nearest-neighbor algorithm a natural fit for \emph{online} $k$-closest POI queries, where the POIs are only revealed at query time. In this case, we need to run both the selection and query phase for each client's request. Except for simple store locators of franchises, online queries are more common than offline queries in interactive map-based services. For example, whenever the set of POIs is obtained from user-defined keywords, we face online queries.

\subsection{Travel Demand Generation}
\label{sec:travel-demand-generation}

A substantially different application in which our nearest-neighbor algorithm can be used is travel demand generation. Here, the problem we consider is computing the number~$T_{vw}$ of trips between each pair~$(v, w)$ of vertices~$v, w \in V$. This problem arises when we want to generate large-scale benchmark data for evaluating transportation algorithms, or when we want to predict mobility flows. This section shows how our nearest-neighbor algorithm can be used to accelerate a recently introduced travel demand generator~\cite{BuchholdSW19a}.

\subparagraph*{Radiation Model.}

The foundation for the aforementioned demand generator is the \emph{radiation model}~\cite{SiminiGMB12}. This model assumes that each vertex $v \in V$ has a nonnegative number $m_v$ of inhabitants and a nonnegative amount $n_v$ of opportunities. We denote by $M$ the total population in $G$ and by $N$ the total number of opportunities in $G$. The mobility flow out of each vertex is proportional to its population. Destination selection is based on the following main idea: Each traveler assigns to all opportunities a fitness or attractiveness value, drawn independently from a common distribution. Then, the traveler selects the closest opportunity with a fitness higher than the traveler's fitness threshold, drawn from the same distribution. The \emph{radiation model with selection}~\cite{SiminiMN13} decreases the probability of selecting an opportunity by a factor of $1 - \lambda$. Intuitively, increasing $\lambda$ increases the expected trip length. In the simplest version, the number of opportunities is approximated by the population, i.e., there are $M$ opportunities in a graph with a population of $M$.

\subparagraph*{Previous Implementations.}

There are two practical implementations~\cite{BuchholdSW19a} of the radiation model. \emph{DRAD} obtains high-quality solutions based on shortest-path distances and \emph{TRAD} obtains high performance but uses geometric distances. Both implementations generate one trip after another. First, they draw the origin~$O$ from a discrete distribution determined by the probability function~$\Pr[O = v] = m_v / M$. Second, they choose the number~$\Ofit$ of opportunities with a fitness higher than the traveler's fitness threshold uniformly at random in $0..N$. Third, they draw the number~$\Osel$ of selectable opportunities from a binomial distribution with $\Ofit$~trials and success probability~$1 - \lambda$. It remains to find the selectable opportunity closest to $O$, given the total number~$\Osel$ of selectable opportunities in $G$. This is realized differently by the two implementations.

DRAD draws the number~$\Oint$ of opportunities that are closer to $O$ than any selectable opportunity from a negative hypergeometric distribution determined by $\Osel$ and $N$, and runs Dijkstra's algorithm from $O$, stopping as soon as $\Oint + 1$~opportunities are visited. The last vertex scanned by the search is the destination of the current trip.

The basic idea of TRAD is to find the selectable opportunity closest to $O$ using a nearest-neighbor query~\cite{FriedmanBF77} in a kd-tree~\cite{Bentley75}. Each node in a kd-tree corresponds to a region of the plane. The region of the root is the whole plane and the leaves correspond to small disjoint blocks partitioning the plane. The query algorithm traverses the kd-tree, starting at the root, and maintaining the number~$\Osel(v)$ of selectable opportunities in the region corresponding to the current node~$v$. Let $\Otot(v)$ be the total number of opportunities in the region of $v$.

When the traversal reaches an interior node~$v$ in the kd-tree, the algorithm draws the number~$\Osel(l)$ of selectable opportunities in the region of the left child~$l$ from a hypergeometric distribution with $\Osel(v)$~draws without replacement from a population of size~$\Otot(v)$ containing $\Otot(l)$~successes. The number~$\Osel(r)$ of selectable opportunities in the region corresponding to the right child~$r$ is set to $\Osel(v) - \Osel(l)$. The algorithm then recurses on the child whose region is closer to $O$, and when control returns, it recurses on the other child. The search is pruned at any vertex~$v$ with $\Osel(v) = 0$, and at any vertex whose region is farther from $O$ than the closest selectable opportunity seen so far.

When the traversal reaches a leaf node~$v$, the algorithm samples $\Osel(v)$~selectable opportunities in the region corresponding to $v$. For each of these opportunities, the algorithm checks whether it improves the closest selectable opportunity seen so far.

\subparagraph*{Our Implementation.}

\begin{algorithm2e}[tb]
  \caption{Recursive procedure for finding the closest selectable opportunity in the subgraph $G_X$, given the number $\Osel(G_X)$ of selectable opportunities in $G_X$.}
  \label{algo:crad}
  \Func{$\mathit{findClosestSelectableOpportunity}(X, \Osel(G_X))$}{
    \If{the recursion threshold is deceeded}{
      sample $\Osel(G_X)$ selectable opportunities in the subgraph $G_X$\;
      \Return\;
    }
    $\langle \Osel(G_{Y_0}), \dots, \Osel(G_{Y_{d - 1}}), \Osel(X) \rangle \gets
        \mathit{multiHypergeomVariate}(
            \Osel(G_X), \langle \Otot(G_{Y_0}), \dots, \Otot(G_{Y_{d - 1}}), \Otot(X) \rangle)$\;
    sample $\Osel(X)$ selectable opportunities in the separator $X$\;
    $C \gets \emptyset$\;
    \ForEach{child $Y$ of $X$}{
      \If{$\Osel(G_Y) > 0$}{
        \eIf{$O \in V(G_Y)$}{
          $C \gets C \cup \{(Y, 0)\}$\;
        }{
          compute the distance $\Dist(O, Y)$ from $O$ to a closest vertex in $G_Y$\;
          $C \gets C \cup \{(Y, \Dist(O, Y))\}$\;
        }
      }
    }
    \ForEach{$(Y, \Dist(O, Y)) \in C$ in ascending order of $\Dist(O, Y)$}{
      \If{$\Dist(O, Y)$ is less than distance to currently closest select.\ opportunity}{
        $\mathit{findClosestSelectableOpportunity}(Y, \Osel(G_Y))$\;
      }
    }
  }
\end{algorithm2e}

We introduce a new implementation of the radiation model, which we call CRAD. Our implementation follows TRAD but uses nearest-neighbor queries in a customizable contraction hierarchy rather than in a kd-tree. In this way, we combine the efficient tree-based sampling approach from TRAD with shortest-path distances. As a result, our implementation obtains high-quality solutions like DRAD, but at much lower cost.

To use our nearest-neighbor algorithm in CRAD, we only need to make slight modifications to the procedure presented in \cref{sec:nearest-neighbors} (see \cref{algo:crad} for the modified procedure). In addition to a node~$X$ in the separator decomposition tree, it now takes the number~$\Osel(G_X)$ of selectable opportunities in $G_X$ as second parameter. At the first call, $X$ is the root of the separator decomposition and $\Osel(G_X)$ is the number~$\Osel$ of selectable opportunities in $G$, obtained as before in DRAD and TRAD. Let $Y_0, \dots, Y_{d - 1}$ be the children of $X$. As the first step, the procedure now distributes the $\Osel$~selectable opportunities in $G_X$ over the subgraphs~$G_{Y_0}, \dots, G_{Y_{d - 1}}$ and the separator~$X$. In contrast to TRAD where the opportunities are distributed among exactly two regions (left and right child), we now have $d + 1$~regions ($d$~children and the separator). Therefore, $\Osel(G_{Y_0}), \dots, \Osel(G_{Y_{d - 1}}), \Osel(X)$ obey a \emph{multivariate} hypergeometric distribution.

We can think of this distribution as drawing $\Osel(G_X)$~balls without replacement from an urn containing $\Otot(G_{Y_i})$~balls of type~$i$ for $i = 0, \dots, d - 1$ and $\Otot(X)$~balls of type~$d$. We obtain $\Osel(G_{Y_i})$~balls of type~$i$ for $i = 0, \dots, d - 1$ and $\Osel(X)$~balls of type~$d$.

After obtaining $\Osel(G_{Y_0}), \dots, \Osel(G_{Y_{d - 1}}), \Osel(X)$, we sample $\Osel(X)$~selectable opportunities in the separator~$X$, and check whether any of them improves the closest selectable opportunity seen so far. Next, we loop over all children~$Y$ of $X$ in the separator decomposition tree. If the subgraph~$G_Y$ contains any selectable opportunities, we add a pair~$(Y, \Dist(O, Y))$ to a set~$C$ (recall that $O$ is the origin vertex of the current trip). The shortest-path distance~$\Dist(O, Y)$ is computed as discussed in \cref{sec:nearest-neighbors}. Finally, we loop over all pairs~$(Y, \Dist(O, Y)) \in C$ in ascending order of distance from the origin. If $\Dist(O, Y)$ is less than the distance to the closest selectable opportunity seen so far, we recurse on $Y$.

\section{Experiments}
\label{sec:experiments}

This section presents a thorough experimental evaluation of both applications. First, we describe our experimental setup, including our benchmark machine, the inputs, and implementation details. Next, we evaluate various closest-POI algorithms, with a focus on their selection and query phases. Finally, we compare CRAD to DRAD and TRAD.

\subsection{Experimental Setup}

Our publicly available code\footnote{\url{https://github.com/vbuchhold/routing-framework}} is written in C++17 and compiled with the GNU compiler~9.3 using optimization level~3. We use 4-heaps~\cite{Johnson75} as priority queues. To ensure a correct implementation, we make extensive use of assertions (disabled during measurements). Our benchmark machine runs openSUSE Leap~15.2 (kernel~5.3.18), and has 192\,GiB of DDR4-2666 RAM and two Intel Xeon Gold~6144 CPUs, each with eight cores clocked at 3.50\,GHz and $8 \times 64$\,KiB of L1, $8 \times 1$\,MiB of L2, and 24.75\,MiB of shared L3~cache.

\subparagraph*{Inputs.}

Our benchmark instance is the road network of Western Europe. The network has a total of \num{18017748} vertices and \num{42560275} edges and was made available by PTV AG for the 9th DIMACS Implementation Challenge~\cite{DemetrescuGJ09}. For the evaluation of the travel demand generators, we use the population grid\footnote{\raggedright\url{https://ec.europa.eu/eurostat/web/gisco/geodata/reference-data/population-distribution-demography/geostat}} made available by Eurostat, the statistical office of the European Union. The grid has a resolution of one kilometer and covers all EU and EFTA member states, as well as the United Kingdom. We follow the approach in \cite{BuchholdSW19a} to assign the grid to the graph. For each inhabitant, we pick a vertex lying in their cell uniformly at random and assign the inhabitant to it. If there is no such vertex, we discard the inhabitant.

\subparagraph*{Implementation Details.}

We use the network dissection algorithm Inertial Flow~\cite{SchildS15} to compute separator decompositions and associated nested dissection orders, with the balance parameter~$b$ set to $3 / 10$. CCH customization uses perfect witness searches~\cite{DibbeltSW16}.

For comparison, we carefully reimplemented the bucket-based nearest-neighbor algorithm by Geisberger~\cite{Geisberger11}, which we call BCH. CH preprocessing is taken from the open-source library RoutingKit\footnote{\url{https://github.com/RoutingKit/RoutingKit}}. Both the forward and reverse CH searches use stall-on-demand~\cite{GeisbergerSSV12}.

The bucket-based nearest-neighbor algorithm can be used as is on CCHs, without further modifications. For better performance, however, we use a tailored version where we replace the Dijkstra-based CH searches used during selection and queries by elimination tree searches. Note that in contrast to CH searches, CCH searches are faster without the stall-on-demand technique. On the other hand, stall-on-demand decreases the bucket sizes. Therefore, we use stall-on-demand only for the reverse searches. We call this version BCCH.

To keep implementation complexity of the demand generators low, we use existing implementations of random variate generation algorithms. The Standard Template Library~(STL) offers the three distribution classes \texttt{uniform\_int\_distribution}, \texttt{binomial\_distribution}, and \texttt{geometric\_distribution}. The STL provides neither a hypergeometric nor a negative hypergeometric distribution. To generate hypergeometric variates, we use the stocc library\footnote{\url{https://www.agner.org/random/}}. Following \cite{BuchholdSW19a}, we approximate negative hypergeometric variates by geometric variates.

\subsection{Online Closest-POI Queries}

We start by comparing our nearest-neighbor algorithm (simply called CCH in this section) to Dijkstra's algorithm, BCH, BCCH, and CRP. Note that the performance of closest-POI algorithms is affected not only by the number of POIs but also by their \emph{distribution}. For example, the set of \emph{all} restaurants may be distributed evenly over the whole network, whereas a certain franchise may operate in a local region. To model this, we follow the methodology used by Delling et al.~\cite{DellingGW11} to evaluate one-to-many algorithms.

To obtain our problem instances, we first pick a center $c$ uniformly at random. We then use Dijkstra to grow a ball $B$ of size $|B|$ centered at $c$. Finally, we pick a POI set $P$ of size $|P|$ from $B$. By varying the parameters $|B|$ and $|P|$, we can model the aforementioned situations. For each combination, we generate 100 POI sets. Each POI set is evaluated with 100 sources picked at random. That is, each data point is an average over \num{10000} queries.

\subparagraph*{Main Results.}

\begin{table}[tb]
  \caption{Performance of different closest-POI algorithms for various POI distributions. For each distribution, we report the time to index a set of POIs (selection time), the space consumed by the index (selection space), and the time to find the $k = 1, 4, 8$ closest POIs (query time). For CRP, we take the figures for the online version from the original publication~\cite{DellingW15}.}
  \label{tab:closest-poi-algorithms}
  \begin{tabular}{lrrSSSrSrrr}
  \toprule
  & \multicolumn{5}{c}{$|P| = 2^{12}, |B| = 2^{20}$} & \multicolumn{5}{c}{$|P| = 2^{14}, |B| = |V|$} \\
  \cmidrule(lr){2-6}\cmidrule(l){7-11}
  & \multicolumn{2}{c}{selection} & \multicolumn{3}{c}{query time [\si{\micro\second}]} & \multicolumn{2}{c}{selection} & \multicolumn{3}{c}{query time [\si{\micro\second}]} \\
  \cmidrule(lr){2-3}\cmidrule(lr){4-6}\cmidrule(lr){7-8}\cmidrule(l){9-11}
  & space & time & \multicolumn{3}{c}{POIs to be reported} & space & time & \multicolumn{3}{c}{POIs to be reported} \\
  algo & [MiB] & [ms] & {$k = 1$} & {$k = 4$} & {$k = 8$} & [MiB] & {[ms]} & $k = 1$ & $k = 4$ & $k = 8$ \\
  \midrule
  Dij  &   -- &  -- & 846210 & 855438 & 873716 &    -- & {--} & 113.4 & 439.3 & 883.7 \\
  BCH  & 72.4 & 134 &     20 &     20 &     21 &  83.6 &  481 &   5.0 &   8.5 &  10.7 \\
  BCCH & 85.5 & 453 &     51 &     52 &     53 & 134.9 & 1753 &   6.0 &   8.8 &  11.1 \\
  CCH  & 68.7 &  21 &   2353 &   3501 &   4629 &  68.7 &   23 & 306.7 & 494.8 & 702.0 \\
  CRP  &   -- &  -- &   {--} &   {--} &   {--} &   0.0 &    8 &    -- & 640.0 &    -- \\
  \bottomrule
\end{tabular}
\end{table}

\Cref{tab:closest-poi-algorithms} shows the performance of different closest-POI algorithms for two POI distributions. We observe that Dijkstra's algorithm has reasonable performance when the POIs are evenly distributed over the whole graph ($|B| = |V|$). In this case, any potential source is relatively close to some POI, and thus the Dijkstra search can always stop early. However, Dijkstra's algorithm is not robust to the POI distribution. When $|B| = 2^{20}$, many potential sources are relatively far from any POI, and the average running times are around one second, too slow for interactive map-based services.

BCH achieves the best (offline) query times for both POI distributions. Note, however, that BCH is no competitor to BCCH, CCH, and CRP, since it operates on \emph{standard} contraction hierarchies, which cannot handle frequent metric updates. We only include BCH in our experiments for comparison with BCCH, since the bucket-based nearest-neighbor algorithm has not been tested on \emph{customizable} contraction hierarchies so far.

Although we tailored the bucket-based algorithm to CCHs, BCCH is still somewhat slower than BCH. This is expected, since CCHs contain more shortcuts and are thus denser than CHs. The slowdown is a factor of about 3.5 for selection. When $|B| = |V|$, BCCH has only slightly higher (offline) query times than BCH, since the queries relax only a few edges. However, BCCH queries are roughly 2.5 times slower than BCH queries when $|B| = 2^{20}$.

We observe that our nearest-neighbor algorithm (simply called CCH in this section) has considerably higher offline query times than BCCH. On the other hand, CCH achieves much faster selection times. For example, when $|P| = 2^{14}$, offline CCH queries are slower by a factor of 51--63 but CCH selection is faster by a factor of 77. Note that although CCH queries are significantly slower than BCCH queries, they are still slightly faster than CRP queries.

Online queries need to run both the selection and query phase for each client's request. Therefore, the time taken by an online query is the sum of the selection and query time. We observe that BCCH is not suitable for online queries. When $|P| = 2^{12}$, BCCH takes half a second to answer an online query, and it takes even 1.8~seconds when $|P| = 2^{14}$. In contrast, CCH takes only about 25~milliseconds for an online query.

\Cref{tab:closest-poi-algorithms} includes various alternative closest-POI algorithms. In addition, it seems natural to adapt existing one-to-many algorithms to the closest-POI problem. Promising candidates that are not based on buckets are CTD \cite{EisnerFHSS11, DellingGW11} and RPHAST \cite{DellingGW11}. However, since CTD and RPHAST selection take more than 100~milliseconds when $|B| = |V|$, online closest-POI queries based on CTD or RPHAST would be at least four times slower than ours.

\subparagraph*{Impact of the POI Distribution.}

Our next experiment considers the impact of the ball size on the performance of the different closest-POI algorithms. \Cref{fig:varying-ball-size} plots selection and (offline) query times for various ball sizes. We omit online query times for clarity. Since the online query times are dominated by the selection times, online query times would closely follow the selection curves. Except for Dijkstra's algorithm, all selection and query times are very robust to the ball size. While all query algorithms benefit from an even distribution of the POIs (for the aforementioned reasons), this effect is most pronounced for Dijkstra.

\begin{figure}[tb]
  \input{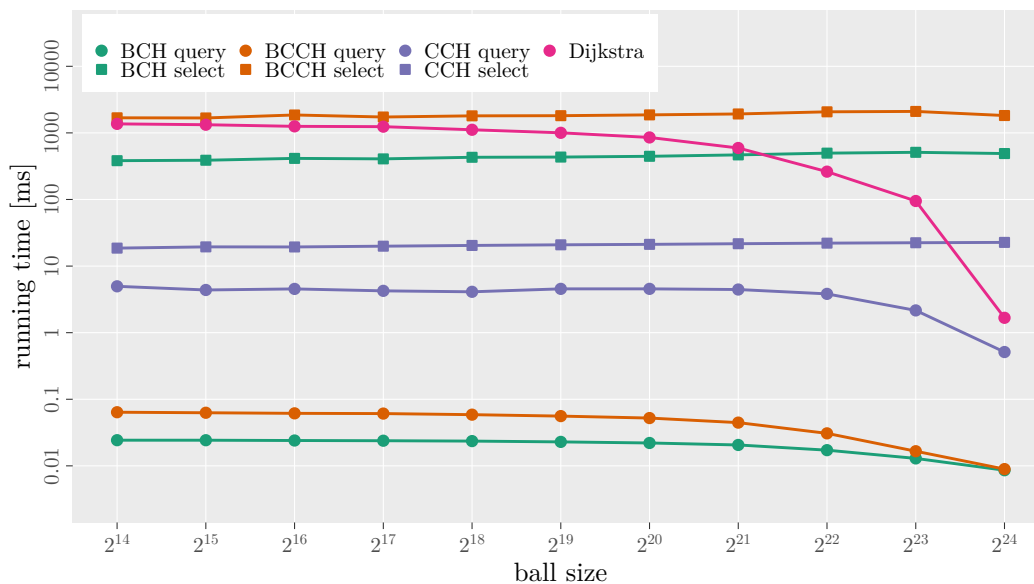}%
  \caption{Selection and query times of various closest-POI algorithms with $|P| = 2^{14}$ POIs picked at random from a ball of varying size $|B|$. Queries find the $k = 4$ closest POIs.}
  \label{fig:varying-ball-size}
\end{figure}

\subparagraph*{Impact of the Number of POIs.}

Next, we evaluate the impact of the number of POIs on the performance of Dijkstra's algorithm, BCH, BCCH, and CCH. \Cref{fig:varying-poi-set-size} plots selection and (offline) query times for various numbers of POIs. As before, online query times would closely follow the selection curves. We observe that the CCH selection time is independent of the number of POIs, whereas the BCCH selection time grows linearly. For $|P| = 2^{14}$, CCH selection is 76~times faster than BCCH selection. The speedup increases to more than three orders of magnitude for $|P| = 2^{18}$, the largest number of POIs tested in our experiment.

Once again, queries tend to become faster as $|P|$ gets larger, since they can stop (in the case of Dijkstra-based searches) or prune (in the case of elimination tree searches) earlier. The exception are CCH queries, which become slower initially. The reason is that for very small values of $|P|$, we do not explore the separator decomposition tree but trigger the base case at the root (which simply finds $|P|$ point-to-point shortest paths by running standard elimination tree queries from the source to each POI).

\begin{figure}[tb]
  \input{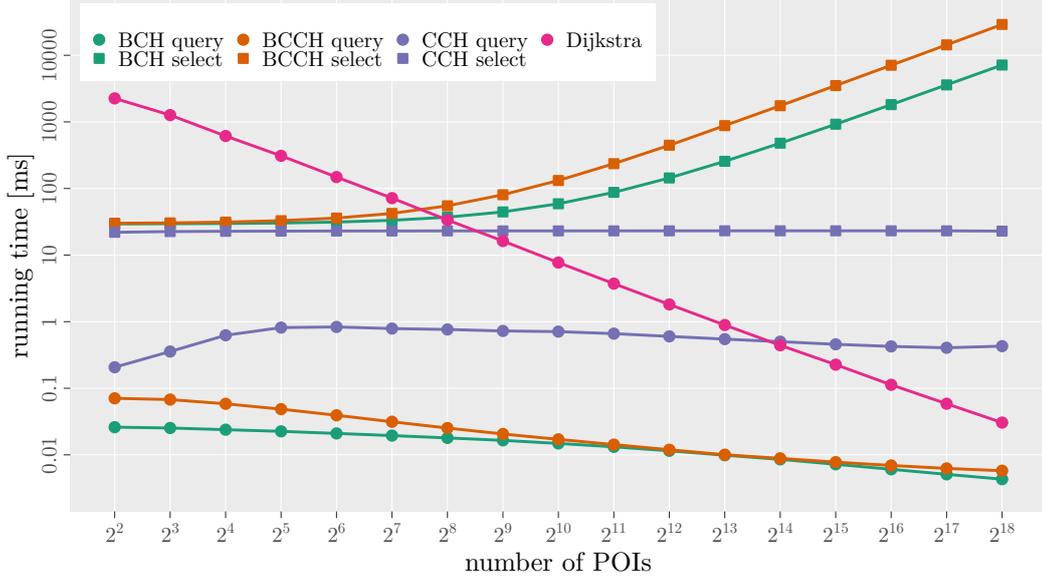}%
  \caption{Selection and query times of various closest-POI algorithms with a varying number $|P|$ of POIs picked at random from a ball of size $|B| = |V|$. Queries find the $k = 4$ closest POIs.}
  \label{fig:varying-poi-set-size}
\end{figure}

\subsection{Travel Demand Generation}

Next, we evaluate CRAD, including a comparison to DRAD and TRAD. Since CRAD uses shortest-path distances rather than geometric distances, it obtains high-quality solutions like DRAD. We verified this experimentally by rerunning the experiments in the original publication~\cite{BuchholdSW19a} for CRAD, using the same instances and methodology. We refer to the original paper for a comparison of the solution quality with shortest-path and geometric distances.

In this work, we focus on the performance of the three implementations. Since DRAD is at its heart a Dijkstra search from the trip's origin to its destination, the performance depends heavily on the expected length of the generated trip (which is controlled by the parameter $\lambda$; see \cref{sec:travel-demand-generation}). In contrast, TRAD and CRAD are robust to the trip length.

\Cref{fig:travel-demand} plots the time to generate a single trip for various values of $\lambda$. Note that a value of $\lambda = 1 - 10^{-4} / 1 = 0.9999$ leads on our instance to an average trip length of 9~minutes, and a value of $\lambda = 1 - 10^{-4} / 100 = 0.999999$ to an average trip length of 72~minutes. Between two data points, the average trip length increases by about 7~minutes. All data points are averages over \num{100000} trip generation executions.

We observe that CRAD outperforms DRAD for each value of $\lambda$ tested. Since TRAD resorts to geometric distances, it still is faster than CRAD by a factor of 28--74. As it obtains worse solutions, however, TRAD is no competitor to CRAD. For an average trip length of about 23~minutes, CRAD gains an order of magnitude over DRAD, and for the largest value of $\lambda$ tested in our experiment, we see a speedup of 59. Note that this increase in speed is quite useful in practice. While travel demand generation does not need to run in real time, its performance should remain reasonable. However, DRAD takes about 7~hours to generate one million one-hour trips. In contrast, CRAD takes less than 10~minutes.

\begin{figure}[tb]
  \centering
  \input{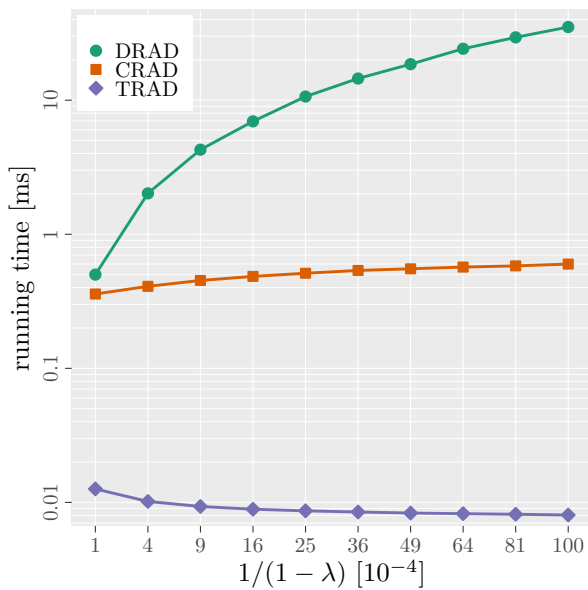}%
  \caption{Time to generate a single trip with different demand generators for various values of $\lambda$.}
  \label{fig:travel-demand}
\end{figure}

\section{Conclusion}
\label{sec:conclusion}

We presented a novel $k$-nearest neighbor algorithm that operates on CCHs. With selection times of about 20 milliseconds and query times of a few milliseconds or less, it is the first nearest-neighbor algorithm operating on CCHs that is fast enough for interactive online queries. Interestingly, our algorithm achieves similar performance as the online nearest-neighbor queries by Delling and Werneck~\cite{DellingW15} within the CRP framework. This confirms that CCHs and CRP are on an equal level and solve many types of problems equally well.

Moreover, we used our nearest-neighbor algorithm to significantly accelerate a recent travel demand generator. We proposed CRAD, a new implementation of the radiation model that combines the advantages of the two previous implementations DRAD and TRAD. CRAD obtains high-quality (shortest-path based) solutions like DRAD, but follows a more efficient tree-based sampling approach like TRAD.

Future work includes accelerating our nearest-neighbor algorithm even further. Note that we compute distances to subgraphs corresponding to the topmost nodes in the separator decomposition more often than distances to subgraphs corresponding to leaves. It would be interesting to see if it pays to precompute the reverse search spaces of the topmost subgraphs. Another possible approach would be to keep frequently used reverse search spaces in an LRU cache. Another interesting project is a parallel version of our algorithm that uses for example task-based parallelism to explore the separator decomposition tree. Finally, it would be interesting to port other point-of-interest algorithms to CCHs, for example best-via queries.

\bibliography{References}

\begin{thebibliography}{10}

\bibitem{AbeywickramaC17}
Tenindra Abeywickrama and Muhammad~Aamir Cheema.
\newblock Efficient landmark-based candidate generation for {kNN} queries on
  road networks.
\newblock In K.~Selçuk Candan, Lei Chen, Torben~Bach Pedersen, Lijun Chang,
  and Wen Hua, editors, {\em Proceedings of the 22nd International Conference
  on Database Systems for Advanced Applications ({DASFAA}'17)}, volume 10178 of
  {\em Lecture Notes in Computer Science}, pages 425--440. Springer, 2017.
\newblock \href {https://doi.org/10.1007/978-3-319-55699-4_26}
  {\path{doi:10.1007/978-3-319-55699-4_26}}.

\bibitem{AbeywickramaCT16}
Tenindra Abeywickrama, Muhammad~Aamir Cheema, and David Taniar.
\newblock K-nearest neighbors on road networks: A journey in experimentation
  and in-memory implementation.
\newblock {\em Proceedings of the VLDB Endowment}, 9(6):492--503, 2016.
\newblock \href {https://doi.org/10.14778/2904121.2904125}
  {\path{doi:10.14778/2904121.2904125}}.

\bibitem{AkibaIKK14}
Takuya Akiba, Yoichi Iwata, Ken ichi Kawarabayashi, and Yuki Kawata.
\newblock Fast shortest-path distance queries on road networks by pruned
  highway labeling.
\newblock In {\em Proceedings of the 16th Meeting on Algorithm Engineering and
  Experiments ({ALENEX}'14)}, pages 147--154. {SIAM}, 2014.
\newblock \href {https://doi.org/10.1137/1.9781611973198.14}
  {\path{doi:10.1137/1.9781611973198.14}}.

\bibitem{BastDGMPSWW16}
Hannah Bast, Daniel Delling, Andrew~V. Goldberg, Matthias Müller-Hannemann,
  Thomas Pajor, Peter Sanders, Dorothea Wagner, and Renato~F. Werneck.
\newblock Route planning in transportation networks.
\newblock In Lasse Kliemann and Peter Sanders, editors, {\em Algorithm
  Engineering: Selected Results and Surveys}, volume 9220 of {\em Lecture Notes
  in Computer Science}, pages 19--80. Springer, 2016.
\newblock \href {https://doi.org/10.1007/978-3-319-49487-6_2}
  {\path{doi:10.1007/978-3-319-49487-6_2}}.

\bibitem{BauerCRW16}
Reinhard Bauer, Tobias Columbus, Ignaz Rutter, and Dorothea Wagner.
\newblock Search-space size in contraction hierarchies.
\newblock {\em Theoretical Computer Science}, 645:112--127, 2016.
\newblock \href {https://doi.org/10.1016/j.tcs.2016.07.003}
  {\path{doi:10.1016/j.tcs.2016.07.003}}.

\bibitem{BaumDPW13}
Moritz Baum, Julian Dibbelt, Thomas Pajor, and Dorothea Wagner.
\newblock Energy-optimal routes for electric vehicles.
\newblock In Craig~A. Knoblock, Peer Kröger, John Krumm, Markus Schneider, and
  Peter Widmayer, editors, {\em Proceedings of the 21st {ACM} {SIGSPATIAL}
  International Conference on Advances in Geographic Information Systems
  ({SIGSPATIAL}'13)}, pages 54--63. {ACM} Press, 2013.
\newblock \href {https://doi.org/10.1145/2525314.2525361}
  {\path{doi:10.1145/2525314.2525361}}.

\bibitem{Bentley75}
Jon~Louis Bentley.
\newblock Multidimensional binary search trees used for associative searching.
\newblock {\em Communications of the ACM}, 18(9):509--517, 1975.
\newblock \href {https://doi.org/10.1145/361002.361007}
  {\path{doi:10.1145/361002.361007}}.

\bibitem{BuchholdSW19a}
Valentin Buchhold, Peter Sanders, and Dorothea Wagner.
\newblock Efficient calculation of microscopic travel demand data with low
  calibration effort.
\newblock In Farnoush Banaei-Kashani, Goce Trajcevski, Ralf~Hartmut Güting,
  Lars Kulik, and Shawn~D. Newsam, editors, {\em Proceedings of the 27th {ACM}
  {SIGSPATIAL} International Conference on Advances in Geographic Information
  Systems ({SIGSPATIAL}'19)}, pages 379--388. {ACM} Press, 2019.
\newblock \href {https://doi.org/10.1145/3347146.3359361}
  {\path{doi:10.1145/3347146.3359361}}.

\bibitem{BuchholdSW19b}
Valentin Buchhold, Peter Sanders, and Dorothea Wagner.
\newblock Real-time traffic assignment using engineered customizable
  contraction hierarchies.
\newblock {\em {ACM} Journal of Experimental Algorithmics},
  24(2):2.4:1--2.4:28, 2019.
\newblock \href {https://doi.org/10.1145/3362693} {\path{doi:10.1145/3362693}}.

\bibitem{BuchholdWZZ20}
Valentin Buchhold, Dorothea Wagner, Tim Zeitz, and Michael Zündorf.
\newblock Customizable contraction hierarchies with turn costs.
\newblock In Dennis Huisman and Christos~D. Zaroliagis, editors, {\em
  Proceedings of the 20th Symposium on Algorithmic Approaches for
  Transportation Modelling, Optimization, and Systems ({ATMOS}'20)}, volume~85
  of {\em {OpenAccess} Series in Informatics ({OASIcs})}, pages 9:1--9:15.
  Schloss Dagstuhl, 2020.
\newblock \href {https://doi.org/10.4230/OASIcs.ATMOS.2020.9}
  {\path{doi:10.4230/OASIcs.ATMOS.2020.9}}.

\bibitem{DellingGNW13}
Daniel Delling, Andrew~V. Goldberg, Andreas Nowatzyk, and Renato~F. Werneck.
\newblock {PHAST}: Hardware-accelerated shortest path trees.
\newblock {\em Journal of Parallel and Distributed Computing}, 73(7):940--952,
  2013.
\newblock \href {https://doi.org/10.1016/j.jpdc.2012.02.007}
  {\path{doi:10.1016/j.jpdc.2012.02.007}}.

\bibitem{DellingGPW17}
Daniel Delling, Andrew~V. Goldberg, Thomas Pajor, and Renato~F. Werneck.
\newblock Customizable route planning in road networks.
\newblock {\em Transportation Science}, 51(2):566--591, 2017.
\newblock \href {https://doi.org/10.1287/trsc.2014.0579}
  {\path{doi:10.1287/trsc.2014.0579}}.

\bibitem{DellingGW11}
Daniel Delling, Andrew~V. Goldberg, and Renato~F. Werneck.
\newblock Faster batched shortest paths in road networks.
\newblock In Alberto Caprara and Spyros~C. Kontogiannis, editors, {\em
  Proceedings of the 11th Workshop on Algorithmic Approaches for Transportation
  Modeling, Optimization, and Systems ({ATMOS}'11)}, volume~20 of {\em
  {OpenAccess} Series in Informatics ({OASIcs})}, pages 52--63. Schloss
  Dagstuhl, 2011.
\newblock \href {https://doi.org/10.4230/OASIcs.ATMOS.2011.52}
  {\path{doi:10.4230/OASIcs.ATMOS.2011.52}}.

\bibitem{DellingW15}
Daniel Delling and Renato~F. Werneck.
\newblock Customizable point-of-interest queries in road networks.
\newblock {\em {IEEE} Transactions on Knowledge and Data Engineering},
  27(3):686--698, 2015.
\newblock \href {https://doi.org/10.1109/TKDE.2014.2345386}
  {\path{doi:10.1109/TKDE.2014.2345386}}.

\bibitem{DemetrescuGJ09}
Camil Demetrescu, Andrew~V. Goldberg, and David~S. Johnson, editors.
\newblock {\em The Shortest Path Problem: Ninth {DIMACS} Implementation
  Challenge}, volume~74 of {\em {DIMACS} Book}.
\newblock American Mathematical Society, 2009.

\bibitem{DibbeltSW16}
Julian Dibbelt, Ben Strasser, and Dorothea Wagner.
\newblock Customizable contraction hierarchies.
\newblock {\em {ACM} Journal of Experimental Algorithmics},
  21(1):1.5:1--1.5:49, 2016.
\newblock \href {https://doi.org/10.1145/2886843} {\path{doi:10.1145/2886843}}.

\bibitem{Dijkstra59}
Edsger~W. Dijkstra.
\newblock A note on two problems in connexion with graphs.
\newblock {\em Numerische Mathematik}, 1:269--271, 1959.

\bibitem{EfentakisP14}
Alexandros Efentakis and Dieter Pfoser.
\newblock {GRASP}. {E}xtending graph separators for the single-source
  shortest-path problem.
\newblock In Andreas~S. Schulz and Dorothea Wagner, editors, {\em Proceedings
  of the 22th Annual European Symposium on Algorithms ({ESA}'14)}, volume 8737
  of {\em Lecture Notes in Computer Science}, pages 358--370. Springer, 2014.
\newblock \href {https://doi.org/10.1007/978-3-662-44777-2_30}
  {\path{doi:10.1007/978-3-662-44777-2_30}}.

\bibitem{EfentakisPV15}
Alexandros Efentakis, Dieter Pfoser, and Yannis Vassiliou.
\newblock {SALT}. {A} unified framework for all shortest-path query variants on
  road networks.
\newblock In {\em Proceedings of the 14th International Symposium on
  Experimental Algorithms ({SEA}'15)}, volume 9125 of {\em Lecture Notes in
  Computer Science}, pages 298--311. Springer, 2015.
\newblock \href {https://doi.org/10.1007/978-3-319-20086-6_23}
  {\path{doi:10.1007/978-3-319-20086-6_23}}.

\bibitem{EisnerFHSS11}
Jochen Eisner, Stefan Funke, Andre Herbst, Andreas Spillner, and Sabine
  Storandt.
\newblock Algorithms for matching and predicting trajectories.
\newblock In Matthias Müller-Hannemann and Renato~F. Werneck, editors, {\em
  Proceedings of the 13th Workshop on Algorithm Engineering and Experiments
  ({ALENEX}'11)}, pages 84--95. {SIAM}, 2011.
\newblock \href {https://doi.org/10.1137/1.9781611972917.9}
  {\path{doi:10.1137/1.9781611972917.9}}.

\bibitem{FriedmanBF77}
Jerome~H. Friedman, Jon~Louis Bentley, and Raphael~A. Finkel.
\newblock An algorithm for finding best matches in logarithmic expected time.
\newblock {\em ACM Transactions on Mathematical Software}, 3(3):209--226, 1977.
\newblock \href {https://doi.org/10.1145/355744.355745}
  {\path{doi:10.1145/355744.355745}}.

\bibitem{Geisberger11}
Robert Geisberger.
\newblock {\em Advanced Route Planning in Transportation Networks}.
\newblock PhD thesis, Karlsruhe Institute of Technology, 2011.
\newblock \href {https://doi.org/10.5445/IR/1000021997}
  {\path{doi:10.5445/IR/1000021997}}.

\bibitem{GeisbergerSSV12}
Robert Geisberger, Peter Sanders, Dominik Schultes, and Christian Vetter.
\newblock Exact routing in large road networks using contraction hierarchies.
\newblock {\em Transportation Science}, 46(3):388--404, 2012.
\newblock \href {https://doi.org/10.1287/trsc.1110.0401}
  {\path{doi:10.1287/trsc.1110.0401}}.

\bibitem{George73}
Alan George.
\newblock Nested dissection of a regular finite element mesh.
\newblock {\em SIAM Journal on Numerical Analysis}, 10(2):345--363, 1973.
\newblock \href {https://doi.org/10.1137/0710032} {\path{doi:10.1137/0710032}}.

\bibitem{GottesburenHUW19}
Lars Gottesbüren, Michael Hamann, Tim~Niklas Uhl, and Dorothea Wagner.
\newblock Faster and better nested dissection orders for customizable
  contraction hierarchies.
\newblock {\em Algorithms}, 12(9):1--20, 2019.
\newblock \href {https://doi.org/10.3390/a12090196}
  {\path{doi:10.3390/a12090196}}.

\bibitem{Guttman84}
Antonin Guttman.
\newblock R-trees: A dynamic index structure for spatial searching.
\newblock {\em {ACM} {SIGMOD} Record}, 14(2):47--57, 1984.
\newblock \href {https://doi.org/10.1145/602259.602266}
  {\path{doi:10.1145/602259.602266}}.

\bibitem{HamannS18}
Michael Hamann and Ben Strasser.
\newblock Graph bisection with pareto optimization.
\newblock {\em {ACM} Journal of Experimental Algorithmics},
  23(1):1.2:1--1.2:34, 2018.
\newblock \href {https://doi.org/10.1145/3173045} {\path{doi:10.1145/3173045}}.

\bibitem{Johnson75}
Donald~B. Johnson.
\newblock Priority queues with update and finding minimum spanning trees.
\newblock {\em Information Processing Letters}, 4(3):53--57, 1975.
\newblock \href {https://doi.org/10.1016/0020-0190(75)90001-0}
  {\path{doi:10.1016/0020-0190(75)90001-0}}.

\bibitem{KnoppSSSW07}
Sebastian Knopp, Peter Sanders, Dominik Schultes, Frank Schulz, and Dorothea
  Wagner.
\newblock Computing many-to-many shortest paths using highway hierarchies.
\newblock In {\em Proceedings of the 9th Workshop on Algorithm Engineering and
  Experiments ({ALENEX}'07)}, pages 36--45. {SIAM}, 2007.
\newblock \href {https://doi.org/10.1137/1.9781611972870.4}
  {\path{doi:10.1137/1.9781611972870.4}}.

\bibitem{KobitzschRS13}
Moritz Kobitzsch, Marcel Radermacher, and Dennis Schieferdecker.
\newblock Evolution and evaluation of the penalty method for alternative
  graphs.
\newblock In Daniele Frigioni and Sebastian Stiller, editors, {\em Proceedings
  of the 13th Workshop on Algorithmic Approaches for Transportation Modeling,
  Optimization, and Systems ({ATMOS}'13)}, volume~33 of {\em {OpenAccess}
  Series in Informatics ({OASIcs})}, pages 94--107. Schloss Dagstuhl, 2013.
\newblock \href {https://doi.org/10.4230/OASIcs.ATMOS.2013.94}
  {\path{doi:10.4230/OASIcs.ATMOS.2013.94}}.

\bibitem{LeeLZ09}
Ken C.~K. Lee, Wang-Chien Lee, and Baihua Zheng.
\newblock Fast object search on road networks.
\newblock In Martin~L. Kersten, Boris Novikov, Jens Teubner, Vladimir Polutin,
  and Stefan Manegold, editors, {\em Proceedings of the 12th International
  Conference on Extending Database Technology ({EDBT}'09)}, pages 1018--1029.
  {ACM} Press, 2009.
\newblock \href {https://doi.org/10.1145/1516360.1516476}
  {\path{doi:10.1145/1516360.1516476}}.

\bibitem{LeeLZT12}
Ken C.~K. Lee, Wang-Chien Lee, Baihua Zheng, and Yuan Tian.
\newblock {ROAD}: A new spatial object search framework for road networks.
\newblock {\em {IEEE} Transactions on Knowledge and Data Engineering},
  24(3):547--560, 2012.
\newblock \href {https://doi.org/10.1109/TKDE.2010.243}
  {\path{doi:10.1109/TKDE.2010.243}}.

\bibitem{PapadiasZMT03}
Dimitris Papadias, Jun Zhang, Nikos Mamoulis, and Yufei Tao.
\newblock Query processing in spatial network databases.
\newblock In Johann-Christoph Freytag, Peter~C. Lockemann, Serge Abiteboul,
  Michael~J. Carey, Patricia~G. Selinger, and Andreas Heuer, editors, {\em
  Proceedings of the 29th International Conference on Very Large Data Bases
  ({VLDB}'03)}, pages 802--813. Morgan Kaufmann, 2003.

\bibitem{SametSA08}
Hanan Samet, Jagan Sankaranarayanan, and Houman Alborzi.
\newblock Scalable network distance browsing in spatial databases.
\newblock In Jason Tsong-Li Wang, editor, {\em Proceedings of the 27th {ACM}
  {SIGMOD} International Conference on Management of Data ({SIGMOD}'08)}, pages
  43--54. {ACM} Press, 2008.
\newblock \href {https://doi.org/10.1145/1376616.1376623}
  {\path{doi:10.1145/1376616.1376623}}.

\bibitem{SankaranarayananAS05}
Jagan Sankaranarayanan, Houman Alborzi, and Hanan Samet.
\newblock Efficient query processing on spatial networks.
\newblock In Cyrus Shahabi and Omar Boucelma, editors, {\em Proceedings of the
  13th {ACM} International Workshop on Geographic Information Systems
  ({GIS}'05)}, pages 200--209. {ACM} Press, 2005.
\newblock \href {https://doi.org/10.1145/1097064.1097093}
  {\path{doi:10.1145/1097064.1097093}}.

\bibitem{SchildS15}
Aaron Schild and Christian Sommer.
\newblock On balanced separators in road networks.
\newblock In Evripidis Bampis, editor, {\em Proceedings of the 14th
  International Symposium on Experimental Algorithms ({SEA}'15)}, volume 9125
  of {\em Lecture Notes in Computer Science}, pages 286--297. Springer, 2015.
\newblock \href {https://doi.org/10.1007/978-3-319-20086-6_22}
  {\path{doi:10.1007/978-3-319-20086-6_22}}.

\bibitem{SiminiGMB12}
Filippo Simini, Marta~C. González, Amos Maritan, and Albert-László
  Barabási.
\newblock A universal model for mobility and migration patterns.
\newblock {\em Nature}, 484(7392):96--100, 2012.
\newblock \href {https://doi.org/10.1038/nature10856}
  {\path{doi:10.1038/nature10856}}.

\bibitem{SiminiMN13}
Filippo Simini, Amos Maritan, and Zoltán Néda.
\newblock Human mobility in a continuum approach.
\newblock {\em {PLOS ONE}}, 8(3):1--8, 2013.
\newblock \href {https://doi.org/10.1371/journal.pone.0060069}
  {\path{doi:10.1371/journal.pone.0060069}}.

\bibitem{ZhongLTZ13}
Ruicheng Zhong, Guoliang Li, Kian-Lee Tan, and Lizhu Zhou.
\newblock G-tree: An efficient index for {KNN} search on road networks.
\newblock In Qi~He, Arun Iyengar, Wolfgang Nejdl, Jian Pei, and Rajeev Rastogi,
  editors, {\em Proceedings of the 22nd ACM International Conference on
  Information and Knowledge Management ({CIKM}'13)}, pages 39--48. {ACM} Press,
  2013.
\newblock \href {https://doi.org/10.1145/2505515.2505749}
  {\path{doi:10.1145/2505515.2505749}}.

\bibitem{ZhongLTZG15}
Ruicheng Zhong, Guoliang Li, Kian-Lee Tan, Lizhu Zhou, and Zhiguo Gong.
\newblock G-tree: An efficient and scalable index for spatial search on road
  networks.
\newblock {\em {IEEE} Transactions on Knowledge and Data Engineering},
  27(8):2175--2189, 2015.
\newblock \href {https://doi.org/10.1109/TKDE.2015.2399306}
  {\path{doi:10.1109/TKDE.2015.2399306}}.

\end{thebibliography}

\end{document}